\documentclass[a4paper,reqno]{amsart}

\usepackage[%
pdfauthor={},%
pdftitle={Generalised Brillouin Zone for Non-Reciprocal Systems},%
]{hyperref}
\usepackage[utf8]{inputenc}
\usepackage[english]{babel}
\usepackage{standalone}
\usepackage{graphicx}
\graphicspath{ {graphics/} }
\usepackage{imakeidx}
\usepackage{setspace}
\usepackage{appendix}
\usepackage{pdfpages}
\makeindex
\usepackage{mathtools}
\usepackage{tabulary}
\usepackage{booktabs}
\usepackage{amsmath}
\usepackage{amssymb}
\usepackage{amsthm}
\usepackage{amsfonts}
\usepackage{tikz-cd}
\usepackage{extarrows}
\usepackage{tabularx}
\usepackage{float}
\usepackage{enumitem}
\usepackage{csquotes}
\restylefloat{table}
\usepackage{bm}

\usepackage{tikz}
\usetikzlibrary{calc,intersections}
\usepackage{tikz-layers}
\usepackage[arrow, matrix, curve]{xy}
\usepackage[final]{microtype}

\makeatletter
\newsavebox\myboxA
\newsavebox\myboxB
\newlength\mylenA

\newcommand*\xoverline[2][0.75]{%
  \sbox{\myboxA}{$\m@th#2$}%
  \setbox\myboxB\null% Phantom box
  \ht\myboxB=\ht\myboxA%
  \dp\myboxB=\dp\myboxA%
  \wd\myboxB=#1\wd\myboxA% Scale phantom
  \sbox\myboxB{$\m@th\overline{\copy\myboxB}$}%  Overlined phantom
  \setlength\mylenA{\the\wd\myboxA}%   calc width diff
  \addtolength\mylenA{-\the\wd\myboxB}%
  \ifdim\wd\myboxB<\wd\myboxA%
    \rlap{\hskip 0.5\mylenA\usebox\myboxB}{\usebox\myboxA}%
  \else
    \hskip -0.5\mylenA\rlap{\usebox\myboxA}{\hskip 0.5\mylenA\usebox\myboxB}%
  \fi}
\makeatother

\newcommand{\mc}[1]{\mathcal{#1}}

\DeclareMathOperator{\interior}{int}

\DeclareMathOperator\Arg{Arg}

\newcommand{\qp}{\alpha+\i\beta}
\newcommand{\qpconj}{\alpha'+\i\beta'}
\newcommand{\associateqp}{e^{-\i\spatialperiod(\alpha+\i\beta)}}

\usepackage[T1]{fontenc} % To switch to the T1 encoding
\usepackage{lmodern} % To switch to Latin Modern
\rmfamily % To load Latin Modern Roman and enable the following NFSS declarations.
\DeclareFontShape{T1}{lmr}{b}{sc}{<->ssub*cmr/bx/sc}{}
\DeclareFontShape{T1}{lmr}{bx}{sc}{<->ssub*cmr/bx/sc}{}
\usepackage{orcidlink}
\usepackage{lipsum}
\usepackage{mathtools}
\usepackage{amsmath}
\usepackage{amssymb}
\usepackage{mathdots}
\usepackage{tikz}
\usetikzlibrary{matrix,positioning,decorations.pathreplacing,calc}
\usetikzlibrary{
    decorations.text,%
    decorations.markings,%
    shadows}
\usepackage{adjustbox}
\usepackage{graphicx}

\usepackage{caption}
\usepackage{subcaption}

\usepackage{dsfont} % for 1 bbmath
\usepackage[format=hang]{caption}

\usepackage[normalem]{ulem}

\newcommand{\abs}[1]{\left\lvert#1\right\rvert}
\usepackage{bm}

\usepackage[
    style=numeric,
    sorting=nty,
    maxnames=99,
    maxalphanames=5,
    natbib=true,
    backend=bibtex,
    sortcites]{biblatex}

\DeclareNameAlias{default}{family-given}

\graphicspath{{figures/}}

\AtEveryBibitem{% Clean up the bibtex rather than editing it
    \clearfield{url}
    \clearfield{issn}
    \clearfield{isbn}
    \clearfield{urldate}

    \ifentrytype{book}{
        \clearfield{pages}}{% 
    }
}

\usepackage{xargs}                      % Use more than one optional parameter in a new commands
\usepackage[prependcaption,textsize=tiny,textwidth=3cm]{todonotes}
\newcommandx{\unsure}[2][1=]{\todo[linecolor=red,backgroundcolor=red!25,bordercolor=red,#1]{#2}}
\newcommandx{\change}[2][1=]{\todo[linecolor=blue,backgroundcolor=blue!25,bordercolor=blue,#1]{#2}}
\newcommandx{\info}[2][1=]{\todo[linecolor=green!30!black,backgroundcolor=green!50,bordercolor=green!30!black,#1]{#2}}
\newcommandx{\improvement}[2][1=]{\todo[linecolor=black,backgroundcolor=black!25,bordercolor=black,#1]{#2}}
\newcommandx{\thiswillnotshow}[2][1=]{\todo[disable,#1]{#2}}
\usepackage{amsthm}
\usepackage[noabbrev,capitalize]{cleveref}
\crefname{proposition}{Proposition}{Propositions}
\crefname{equation}{}{}

\newtheorem{theorem}{Theorem}[section]
\newtheorem{lemma}[theorem]{Lemma}
\newtheorem{proposition}[theorem]{Proposition}

\theoremstyle{definition}
\newtheorem{definition}[theorem]{Definition}

\newtheorem{remark}[theorem]{Remark}

\crefname{assumption}{Assumption}{Assumptions}
\crefname{definition}{Definition}{Definitions}
\crefname{corollary}{Corollary}{Corollaries}
\crefname{enumi}{item}{items}
\DeclareMathOperator{\N}{\mathbb{N}}
\DeclareMathOperator{\Z}{\mathbb{Z}}

\DeclareMathOperator{\R}{\mathbb{R}}
\DeclareMathOperator{\C}{\mathbb{C}}

\renewcommand{\i}{\mathbf{i}}

\renewcommand{\bar}[1]{\overline{#1}}

\renewcommand{\qedsymbol}{$\blacksquare$}

\usepackage{fancyhdr}

\fancypagestyle{plain}{%
    \fancyhead[C]{}
    \fancyfoot[C]{}
    
    }
\fancypagestyle{nosection}{%
    \fancyhead[CE]{}
    \fancyhead[CO]{}
    \fancyfoot[CE,CO]{\thepage}
}
\renewcommand{\epsilon}{\varepsilon}
\DeclareMathOperator{\dd}{d\!}

\renewcommand{\i}{\mathbf{i}}

\usepackage{tcolorbox}
\usetikzlibrary{tikzmark,calc}

\DeclareMathOperator*{\spatialperiod}{\mathsf{L}}

\usepackage{upgreek}
\title{Generalised Brillouin Zone for Non-Reciprocal Systems}
\usepackage{orcidlink}
\begin{document}
%\includepdf[]{notes.pdf}

 \author[H. Ammari]{Habib Ammari \orcidlink{0000-0001-7278-4877}}
\address{\parbox{\linewidth}{Habib Ammari\\
 ETH Z\"urich, Department of Mathematics, Rämistrasse 101, 8092 Z\"urich, Switzerland, \href{http://orcid.org/0000-0001-7278-4877}{orcid.org/0000-0001-7278-4877}}}
\email{habib.ammari@math.ethz.ch}
\thanks{}

\author[S. Barandun]{Silvio Barandun \orcidlink{0000-0003-1499-4352}}
 \address{\parbox{\linewidth}{Silvio Barandun\\
 ETH Z\"urich, Department of Mathematics, Rämistrasse 101, 8092 Z\"urich, Switzerland, \href{http://orcid.org/0000-0003-1499-4352}{orcid.org/0000-0003-1499-4352}}}
 \email{silvio.barandun@sam.math.ethz.ch}

\author[P. Liu]{Ping Liu \orcidlink{0000-0002-7857-7040}}
 \address{\parbox{\linewidth}{Ping Liu\\
  School of Mathematical Sciences, Zhejiang University, Zhejiang 310058, China, \href{http://orcid.org/0000-0002-7857-7040}{orcid.org/0000-0002-7857-7040}}}
\email{pingliu@zju.edu.cn}

\author[A. Uhlmann]{Alexander Uhlmann \orcidlink{0009-0002-0426-6407}}
 \address{\parbox{\linewidth}{Ping Liu\\
 ETH Z\"urich, Department of Mathematics, Rämistrasse 101, 8092 Z\"urich, Switzerland, \href{http://orcid.org/0009-0002-0426-6407}{orcid.org/0009-0002-0426-6407}}}
\email{alexander.uhlmann@sam.math.ethz.ch}

\begin{abstract} {Recently, it has been observed that the Floquet-Bloch transform with real quasiperiodicities fails to capture the spectral properties of non-reciprocal systems. The aim of this paper is to introduce the notion of a generalised Brillouin zone by allowing the quasiperiodicities to be complex in order to rectify this. It is proved that this shift of the Brillouin zone into the complex plane accounts for the unidirectional spatial decay of the eigenmodes and leads to correct spectral convergence properties. The results in this paper clarify and prove rigorously how the spectral properties of a finite structure are associated with those of the corresponding semi-infinitely or infinitely periodic lattices and give explicit characterisations of how to extend the Hermitian theory to non-reciprocal settings. Based on our theory, we characterise the generalised Brillouin zone for both open boundary conditions and periodic boundary conditions. Our results are consistent with the physical literature and give explicit generalisations to the $k$-Toeplitz matrix cases.
}
\end{abstract}
\maketitle
\bigskip


\bigskip


\section{Introduction}

In spatially periodic Hermitian systems, such as subwavelength resonator systems in classical wave physics or electronic systems in condensed matter physics, the band structure of the spectrum of the underlying periodic differential operator is described by the band theory in terms of the Floquet-Bloch wave functions.  The frequency or energy spectrum is computed over the Brillouin zone (the set of quasiperiodicities which are real) and consists in general of bands separated by gaps. 

A fundamental question is to consider what happens when the number of subwavelength resonators or atoms gradually increases generating an infinitely periodic lattice (to form a chain, a screen or a crystal). In electronic structures, it is known that the addition of every new atom adds one more energy level, and, in the limit when the number of atoms goes to infinity, we get continuous bands.
A similar result holds in subwavelength wave physics. In \cite{ammariSpectralConvergenceLarge2023}, it is shown that the subwavelength resonant frequencies of a system of coupled resonators in a truncated periodic lattice converge to the essential spectrum of the corresponding infinite lattice. Moreover, the (discrete) density of states for the finite system converges in distribution to the (continuous) density of states of the infinite one. This is achieved by proving a weak convergence of the finite capacitance matrix (which provides a discrete approximation of the spectrum of the differential operator) to the corresponding (translationally invariant) Toeplitz matrix of the infinite structure.

In this paper, we consider non-reciprocal (non-Hermitian) systems such as systems of subwavelength resonators or electronic systems both with imaginary gauge potentials. 
Non-Hermitian systems have been observed to display a variety of scattering and resonance behaviours that are not possible in Hermitian systems 
\cite{el-ganainy.makris.ea2018Nonhermitian,ashida.gong.ea2020NonHermitian}. There are three widely studied classes of non-Hermitian systems: (i)  parity-time symmetric systems where the material parameters take complex values, (ii) systems where an imaginary gauge potential (in the form of a first-order directional derivative) is added to break Hermiticity, and (iii) time-modulated systems where the material parameters depend periodically in time.  While the Green functions associated with systems in class (i) are symmetric with respect to their spatial variables (source and receiver points), systems in classes (ii) and (iii) may possess non-symmetric Green's functions \cite{ammari.davies.ea2021Functional}. This manifests itself by non-reciprocal wave propagation or non-symmetric band structures. Here, we focus our attention on systems in class (ii) and show that the Floquet-Bloch transform with real quasiperiodicities fails to capture the spectral properties of non-reciprocal systems. In order to rectify this, we introduce the notion of generalised Brillouin zone by allowing the quasiperiodicities to be complex. We prove that this generalisation of the Brillouin zone into the complex plane accounts for the unidirectional spatial decay of the eigenvectors.  
We refer the reader to \cite{borisov.fedotov2022Bloch, PhysRevLett.125.226402,PhysRevLett.121.086803,PhysRevLett.123.246801,10.1093/ptep/ptaa100,MR4719028,okuma.sato2023Nonhermitian,okumaNonHermitianTopologicalPhenomena2023} for some earlier formal results obtained by the physics community in this direction mostly on one-dimensional non-Hermitian tight-binding models described by tridiagonal Toeplitz matrices.

In this work, we consider the more general setting of polymer systems and study three subclasses of such systems in class (ii): finite, semi-infinite, and infinite. These physical systems are respectively modelled by tridiagonal $k$-Toeplitz matrices, tridiagonal $k$-Toeplitz operators, and tridiagonal $k$-Laurent operators. It is worth emphasising that Hermitian systems are insensitive to boundary conditions, causing the semi-infinite and infinite spectrum to coincide and the finite system to converge to that limit, both under open or periodic boundary conditions. However, in the non-reciprocal setting, a phenomenon known as the \enquote{non-Hermitian skin effect} occurs in the presence of imaginary gauge potentials and yields exponentially localised modes \cite{ammari.barandun.ea2024Mathematical}. This significantly enriches the behaviour of the underlying differential operators and causes their spectra to diverge. Classical Floquet-Bloch theory cannot capture the localised modes and exponentially converging pseudoeigenvalues break spectral convergence, as the spectral limit of the finite Toeplitz matrix is no longer the Toeplitz operator.

The main idea to rectify these issues is to extend the classical Brillouin zone into the complex plane to model non-reciprocity. Due to the non-reciprocal sensitivity to boundary conditions, the appropriate generalisation will depend on the boundary conditions and the limit of interest. With this approach, we find explicit characterisations of the generalised Brillouin zones.

\begin{description}
    \item[\cref{thm: GFBT}] shows that for the infinite Toeplitz operator case, the appropriate generalisation of the Brillouin zone is of higher dimension to capture the range of allowable decay rates, while \cref{prop:eigenvectors} allows for the construction of the corresponding eigenvectors;
    \item[\cref{thm:obcconv,thm:pbcconv}] show that spectral convergence in non-reciprocal finite systems can be restored using an appropriate shift of the Brillouin zone into the complex plane.
\end{description}
Our results agree with \cite{okumaNonHermitianTopologicalPhenomena2023} and apply to polymer systems characterised by tridiagonal $k$-Toeplitz matrices. It is expected that these results will generalise to time-modulated systems where non-Hermitian skin effects arise as studied in \cite{matsushima2024}. It is worth emphasising that in the absence of such skin effects, the use of the standard Brillouin zone leads to the spectral properties of large but finite time-modulated systems by approximating them by the corresponding infinite periodic systems. Furthermore, the presented theory of a generalised Brillouin zone coincides with the standard one when applied to reciprocal systems, making it a legitimate generalisation. 

The paper is organised as follows. Section \ref{sec:sota} introduces the problem formulation, providing background on Toeplitz theory, an overview of Floquet-Bloch theory, spectral convergence in the Hermitian setting, and a simple non-reciprocal model to illuminate the issues of the Hermitian theory in this setting.
Section \ref{sec:topgbz} focuses on restoring Floquet-Bloch theory for the infinite Toeplitz operators by introducing a generalised Brillouin zone. We introduce the concept of the non-reciprocity rate and prove our main theorem on the spectral decomposition of Toeplitz operators using the generalised Brillouin zone.
Section \ref{sec:3limits} examines three spectral limits which no longer coincide in non-reciprocal settings: the open limit (corresponding to open boundary conditions), the periodic limit (corresponding to periodic boundary conditions), and the pseudospectral limit. We show how these limits differ and relate to each other and characterise the appropriate generalised Brillouin zone for each of them. 

\section{Problem formulation}\label{sec:sota} 
\subsection{Toeplitz theory}
In this work, we consider tridiagonal $k$-Toeplitz matrices and operators.
\begin{definition}[$k$-Toeplitz operators and $k$-Laurent operators]\label{def:kToeplitzLaurent}
    A \emph{$k$-Toeplitz} operator is an infinite matrix of the form
    \begin{align}
    A = \begin{pmatrix}
        A_0 & A_{-1} & A_{-2} & \cdots\\
        A_1 & A_0 & A_{-1} & \cdots\\
        A_2 & A_1 & A_0 & \cdots\\
        \vdots & \ddots & \ddots & \ddots
    \end{pmatrix}
    \label{eq: ktoeplitz op}
    \end{align}
    for a sequence $(A_j)_{j\in\Z}\subset \C^{k\times k}$ of $k\times k$ matrices. We may consider this as an operator on $\ell^2(\C)$. Similarly, a $k$-Laurent operator is an operator of the form \begin{align}
    A = \begin{pmatrix}
        \ddots & \ddots & \ddots & \vdots & \cdots\\
         \cdots & A_1 & A_0 & A_{-1} & \cdots\\
        \cdots &A_2 & A_1 & A_0 & \cdots\\
        \vdots & \vdots & \ddots & \ddots & \ddots
    \end{pmatrix}.
    \label{eq: klaurent op}
    \end{align}
\end{definition}
\begin{definition}[Symbol of a $k$-Toeplitz operator]
    Let $A$ be a $k$-Toeplitz (or Laurent) operator as in \cref{def:kToeplitzLaurent}. Then the unique function $a\in L^\infty(S^1,\C^{k\times k})$ such that
    \begin{align}
    A_k = \frac{1}{2\pi}\int_0^{2\pi}a(e^{\i \theta})e^{-\i k\theta}\dd \theta
        \label{eq: symbol of k-toeplitz}
    \end{align}
    is called the \emph{symbol} of $A$ and we write $A=T(a)$ or $A=L(a)$ in the case of Laurent operators.
\end{definition}
In this work, we will only consider continuous symbols. This is enough for our considerations, but generalisations are possible.

\begin{definition}[$k$-Toeplitz matrix]
    For $n\geq 1$, we define the projections
    \begin{align}
        P_n:\ell^2(\N,\C)&\to\ell^2(\N,\C)\nonumber\\
        (x_1,x_2,\dots)&\mapsto(x_1,\dots x_n,0,0,\dots).
        \label{eq:  definiton of Pn}
    \end{align}
    The \emph{$k$-Toeplitz matrix} of order $mk$ for $m\in\N$ associated to the symbol $a\in L^\infty(S^1,\C^{k\times k})$ is given by
    $$
    T_{m\times k}(a)\coloneqq P_{mk}T(a)P_{mk},
    $$
    and can be identified as an $mk\times mk$ matrix.
\end{definition}

A $k$-Toeplitz matrix or operator $M$ which satisfies
$$
M_{i,j} = 0 \quad \text{for}\quad \vert i - j \vert > 1
$$
is said to be \emph{tridiagonal}. It is easy to check that the symbol $a(z)$ of a tridiagonal $k$-Toeplitz operator
\[
T(a) = \begin{pmatrix}
    a_1 & b_1 & &&& \\
            c_1    & a_2 & b_2 & & &    \\
                  & c_2 & \ddots & \ddots &&      \\
                & & \ddots & \ddots& b_{k-1} &       \\
               & & & c_{k-1} & a_{k} & b_{k}      \\
            &  &   &    &   c_{k}     & a_1 & b_1 \\
             &  &   &    &   & \ddots & \ddots& \ddots 
\end{pmatrix}
\]
is of the form (see \cite{ammari.barandun.ea2024Spectra})
\begin{align}
a:z \mapsto
\begin{pmatrix}
    a_1 & b_1 &0  &  \cdots      &  0 &     c_k z \\
            c_1    & a_2 & b_2 & & & 0    \\
            0      & c_2 & \ddots & \ddots &&   \vdots    \\
            \vdots    & & \ddots & \ddots& b_{k-2} & 0      \\
            0     & & & c_{k-2} & a_{k-1} & b_{k-1}      \\
            b_k z^{-1} & 0 &  \cdots  &  0  &   c_{k-1}     & a_k 
\end{pmatrix} \in \C^{k \times k}.
\end{align}

Given an eigenpair $(\lambda, \bm v)$ of a Toeplitz matrix $T_{m\times 1}(a)$, it is well-known \cite{trefethen.embree2005Spectra} that there is a link between the winding of the symbol $a$ around $\lambda$ and the properties of $\bm v$. Specifically, denoting by $I(f,z_0)\coloneqq \frac{1}{2\pi \i}\int_f (\xi - z_0)^{-1} \dd \xi$ the winding number of a function $f:S^1\to\C$,  we know that there exists an $M>0$ such that
\begin{align}\label{eq: exp decay based on symbol}
    \frac{\vert v^{(j)}\vert }{\max_j \vert v^{(j)}\vert } \leq \begin{dcases}
        M^{-j}, & I(a,\lambda) < 0,\\
        M^{j}, & I(a,\lambda) > 0.
    \end{dcases}
\end{align}
This result has been generalised to $k$-Toeplitz matrices in \cite{ammari.barandun.ea2024Spectra}.

We call a symbol \emph{collapsed} if for all $\lambda\in \C$ it holds that $I(\det(a-\lambda),0)=0$. This corresponds to the fact that the curves traced out by the eigenvalues of the symbol $S^1\ni z \mapsto\sigma(a(z))$ do not generate winding regions. Relating this to \eqref{eq: exp decay based on symbol}, there is no exponential behaviour of the eigenvectors of a collapsed symbol. In the tridiagonal setting this is easy to identify.
\begin{proposition}
    Let $T(a)$ be a symmetric or Hermitian tridiagonal Toeplitz operator. Then the symbol $a(z)$ is collapsed.
\end{proposition}
\begin{proof}
    For $z\in S^1$ the result then follows immediately from the fact that $ (a(\overline{z}))^{\top}=a(z)$ implying $\sigma(a(\overline{z}))=\sigma(a(z))$ in the symmetric case and $(a(z))^\dagger = a(z)$ implying $\sigma(a(z))\subset \R$ in the Hermitian case.
\end{proof}

\subsection{Physical systems and their mathematical models}

Various physical systems are modelled through Toeplitz matrices and operators and variations thereof. These include systems of subwavelength resonators in (classical) one-dimensional wave physics \cite{ammari.barandun.ea2023Exponentially, ammari.barandun.ea2024Mathematical} and the tight-binding model with nearest neighbour approximation in condensed matter theory \cite{okuma.sato2023Nonhermitian,osti_377103,osti_377103complex,cpa.21735,thouless_rev}. These models are constituted by resonators or particles all of which we will call sites in this work. We assume that the interactions between the sites repeat periodically with period $k$, so that if the interactions are all the same $k=1$ holds. We denote by $\spatialperiod$ the spatial period of recurrence.

In all of these examples, the following modelling applies:
\begin{description}
    \item[Finite systems] are constituted by a finite number of sites. These are modelled by tridiagonal $k$-Toeplitz matrices;
    \item[Semi-infinite systems] are constituted by an infinite number of sites but only in one direction from a fixed origin. These are modelled by tridiagonal $k$-Toeplitz operators;
    \item[Infinite systems] are constituted by an infinite number of sites where no point is a privileged choice of origin. These are modelled by tridiagonal $k$-Laurent operators.
\end{description}
In the literature, these three cases are also known as open boundary conditions, semi-infinite boundary conditions, and periodic boundary conditions \cite{okuma.sato2023Nonhermitian}.

\subsection{Floquet-Bloch theory in the Hermitian case}

Floquet--Bloch theory is the proper tool to analyse periodic systems in the Hermitian case, especially because of the Floquet theorem relating the spectra of the infinite operator to the spectra of the single bands. Here, the Hermiticity of the system is reflected in the Hermiticiy of the matrices and operators, \emph{i.e.}, $M=M^*\coloneqq\bar{M}^\top$, where the superscript $\top$ denotes the transpose.

One may quickly notice that when studying a tridiagonal system associated to the Laurent operator $L(a)$ and denoting by $\alpha$ the quasiperiodicity, the operator associated to the Floquet-transformed system is simply given by the symbol 
$a(e^{-\i\alpha \spatialperiod})$. Using \cite[Theorem 2.8]{ammari.barandun.ea2024Spectra}, we can find that
\begin{align}\label{eq:ClassicalFBT}
    \sigma(L(a)) = \bigcup_{\alpha \in Y^*} \sigma(a(e^{-\i\alpha \spatialperiod})),
\end{align}
where $Y^*\coloneqq [-\pi/\spatialperiod,\pi/\spatialperiod)$ is the first \emph{Brillouin zone}. This exactly mirrors the Floquet-Bloch decomposition of the spectrum for periodic self-adjoint elliptic operators. 

Combining the Bauer–Fike theorem together with \cite[Corollary 6.16]{bottcher.silbermann1999Introduction} shows that for these Hermitian systems, the spectrum of the finite system converges to the spectrum of the infinite one, meaning that 
\begin{align}\label{eq:ClassicalConv}
    \sigma(T_{mk}(a))\xrightarrow{m\to\infty}\sigma (L(a))
\end{align}
in the Hausdorff sense. On the other side, the Hermiticity of the symbol immediately implies that
\begin{align*}
    \sigma(T(a))=\sigma(L(a)).
\end{align*}

\subsection{The non-reciprocal case}
We now shift our focus to the case of non-reciprocal systems, that is the case where the matrices and operators we are working with are no longer Hermitian. Non-reciprocal systems are peculiar for having eigenmodes which are condensed on one edge of the system \cite{PhysRevLett.121.086803,ammari.barandun.ea2024Mathematical} and therefore present a privileged choice of origin, making a semi-infinite system the natural corresponding physical structure.

We consider the following prototypical example:
\begin{align}\label{eq: example symbol}
    a(z)=\begin{pmatrix}
        0 & -2 + -\frac{z}{10} \\ -\frac{9}{10} + \frac{1}{z} & 0
    \end{pmatrix},
\end{align}
and look at the spectra plotted in \cref{fig:different spectra}.
\begin{figure}[h]
    \centering
    \includegraphics[width=0.8\linewidth]{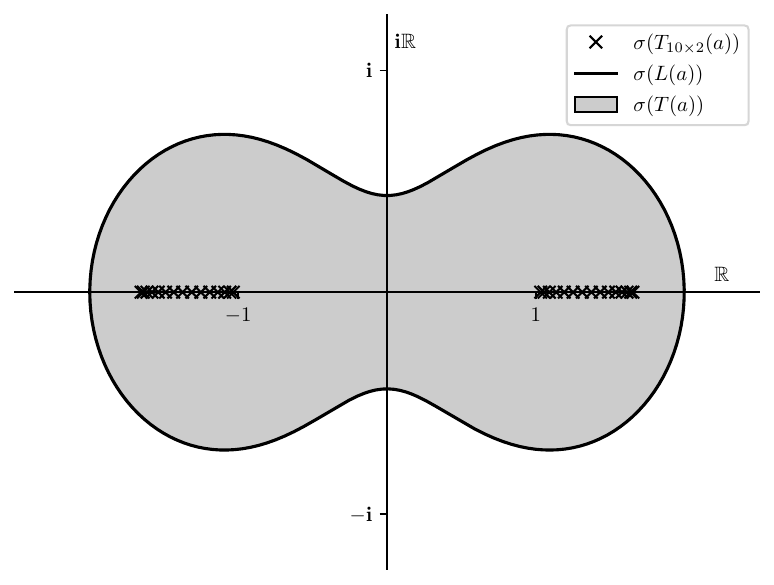}
    \caption{Spectra of the different mathematical objects related to the symbol $a$ from \eqref{eq: example symbol}.}
    \label{fig:different spectra}
\end{figure}

We observe the following:
\begin{enumerate}
\item[(i)] The symbol $a(z)$ is no longer collapsed and now has nonempty interior. This will prove to be the crucial difference between the non-reciprocal and the reciprocal cases, as in the non-reciprocal setting the spectra $\sigma (L(a))$ and $\sigma (T(a))$ do not agree anymore.
\item[(ii)] $\sigma(L(a)) = \bigcup_{\alpha \in Y^*} \sigma(a(e^{-\i\alpha \spatialperiod}))$ still holds also in the non-reciprocal case. Nevertheless, non-reciprocal systems have a privileged choice of origin as they present exponential decay in their eigenmodes. One would wish that the Floquet-Bloch decomposition could model this decay and would cover $\sigma(T(a))$ and not only $\sigma(L(a))$.
\item[(iii)] The convergence $\sigma(T_{mk}(a))\xrightarrow{m\to\infty}\sigma (L(a))$ does not hold anymore and neither does $\sigma(T_{mk}(a))\xrightarrow{m\to\infty}\sigma (T(a))$. The spectrum of $T_{mk}(a)$ is purely real while the ones of $L(a)$ and of $T(a)$ have non-trivial imaginary parts.
\end{enumerate}

In the following sections, we will address the issues (ii) and (iii) above and resolve both of them.

\section{Toeplitz operator and generalised Brillouin zone} \label{sec:topgbz}
As seen in \cref{sec:sota}, the classical Floquet-Bloch transform with real quasiperiodicities $\alpha \in Y^*$ fails to capture non-reciprocal decay and only covers the Laurent operator $L(a)\subsetneq T(a)$ (identity (\ref{eq:ClassicalFBT}) still holds). 
In order to rectify this, we extend the allowable quasiperiodicities into the complex plane. This is a natural extension. Indeed, considering the quasiperiodicity condition
$$
u(x + \spatialperiod) = e^{\i \spatialperiod \alpha} u(x),
$$
one immediately notices that decaying functions (as are the eigenmodes of non-reciprocal systems) cannot be described through this relation for $\alpha\in\R$. This would, however, be the case if we allow $\alpha \in \C$.
\begin{definition}
    For a tridiagonal $k$-Toeplitz operator with symbol
    \begin{align} \label{eq: k-Toeplitz base symbol}
        a:z\mapsto \begin{pmatrix}
    a_1 & b_1 &0  &  \cdots      &  0 &     c_k z \\
            c_1    & a_2 & b_2 & & & 0    \\
            0      & c_2 & \ddots & \ddots &&   \vdots    \\
            \vdots    & & \ddots & \ddots& b_{k-2} & 0      \\
            0     & & & c_{k-2} & a_{k-1} & b_{k-1}      \\
            b_k z^{-1} & 0 &  \cdots  &  0  &   c_{k-1}     & a_k 
\end{pmatrix},
    \end{align}
    with non-zero off-diagonal entries we define the \emph{non-reciprocity rate} as 
    \begin{equation}\label{eq:deltadef}
        \Delta = \ln \prod_{j=1}^k \abs{\frac{b_j}{c_j}}.
    \end{equation}
    Furthermore, we define the \emph{generalised Brillouin zone} to be
    \begin{align}
        \mc B = \bigg\{ \qp \mid \alpha\in [-\pi/\spatialperiod,\pi/\spatialperiod), \beta \in [0,\Delta/\spatialperiod] \bigg\}, 
        \label{eq: GBZ}
    \end{align}
    where $\spatialperiod$ denotes the physical length of the unit cell. 
\end{definition}
We aim at showing that this expansion of the Brillouin zone allows us to reinstate the Floquet-Bloch theorem in a physical sense, which we encompass in the following theorem. 
\begin{theorem}\label{thm: GFBT}
    Consider a tridiagonal $k$-Toeplitz operator with symbol $a$ as in \eqref{eq: k-Toeplitz base symbol} and with non-zero off-diagonal entries and let $\mc B$ be the generalised Brillouin zone from \eqref{eq: GBZ}. Then,
    \begin{align}\label{equ:floquetidentitynonhermitian}
        \sigma(T(a)) = \bigcup_{\qp \in \mc B} \sigma(a(\associateqp)),
    \end{align}
    up to at most $(k-1)$ points which may be in $\sigma(T(a))$ but not in $\bigcup_{\qp \in \mc B} \sigma(a(\associateqp))$. Furthermore, for every $\lambda \in \sigma(T(a))$, the Brillouin zone $\mc B$ contains exactly two corresponding quasiperiodicities 
    \begin{align*}
        \alpha+\i\beta &\in [-\pi/\spatialperiod,\pi/\spatialperiod)+\i[0,\Delta/(2\spatialperiod)] \text{  and }\\ \alpha'+\i\beta' =  (-\zeta/\spatialperiod-\alpha) + \i(\Delta/\spatialperiod-\beta) &\in [-\pi/\spatialperiod,\pi/\spatialperiod)+\i[\Delta/(2\spatialperiod),\Delta/\spatialperiod]
    \end{align*}
    such that $\lambda\in \sigma(a(e^{-\i\spatialperiod(\qp)}) = \sigma(a(e^{-\i\spatialperiod(\alpha'+\i\beta')}))$. Here, $\zeta$ denotes the shift $\zeta \coloneqq \Arg (\prod_{i=1}^k \frac{b_i}{c_i})$. 
\end{theorem}
\begin{remark}\hfill
\begin{itemize}
    \item From now on, we will take $\spatialperiod=1$ without any loss of generality. To reintroduce $\spatialperiod$, we simply rescale $\mc B$ by $1/\spatialperiod$ and all occurrences of $\qp$ in the formulas by $\spatialperiod$.
    \item We will also use $\qp$ and $z=e^{-\i(\qp)}$ interchangeably and refer to them as \emph{associated}.
    \item For a given quasiperiodicity $\qp\in \mc B$ with $\beta\in[0,\Delta/(2\spatialperiod)]$, we call $\alpha'+\i\beta'$ with $\alpha' = \zeta/\spatialperiod-\alpha$ and $\beta' = \Delta/\spatialperiod-\beta$ the \emph{conjugate quasiperiodicities}.
    \item In the above definition, we have assumed that $\Delta>0$. However, this is not necessarily the case. If $\Delta<0$, then the eigenmodes\footnote{While these eigenmodes $\bm u$ are eigenmodes in the symbolic sense $(T(a)-\lambda I) \bm u = 0$, they no longer lie in $\ell^2$ due to their exponential growth.} of $T(a)$ will turn out to be exponentially growing and we observe a \emph{negative} decay parameter $\beta$. We can then let $\beta \in [\Delta/\spatialperiod,0]$ and all of the arguments below work analogously. 
    \item As $\alpha\mapsto e^{-\i(\alpha+\i\beta)}$ is periodic in $\alpha$ with period $2\pi$ we consider equality with respect to $\alpha$ modulo $2\pi$ and choose the convention $\alpha\in [-\pi, \pi)$. 
    \item For reciprocal systems, \emph{i.e.}, for $\Delta=0$,  the generalised Brillouin zone reduces to the standard Brillouin zone, effectively making $\mc B$ an extension of $Y^*$.
    \end{itemize}
\end{remark}

To prove \cref{thm: GFBT}, we will need some intermediate results from \cite{ammari.barandun.ea2024Spectra} providing insights into the spectrum of $T(a)$.
\begin{proposition}\label{prop: char of spectrum k toeplitz}
    Consider a symbol $a(z)$ as in \eqref{eq: k-Toeplitz base symbol}. Then,
    \begin{equation*}
    \sigma_{\det}(a) \cup \sigma_{\mathrm{wind}}(a) \subseteq \sigma(T(a)) \subseteq \sigma_{\det}(a) \cup \sigma_{\mathrm{wind}}(a) \cup \sigma(B_0),
    \end{equation*}
    where 
    \begin{equation*}
        \sigma_{\det}(a) \coloneqq \left\{ \lambda \in \mathbb{C} : \det(a(z) - \lambda) = 0, \ \exists z \in S^1\right\},
    \end{equation*} 
    \begin{equation*}
        \sigma_{\mathrm{wind}}(a) \coloneqq \left\{ \lambda \in \mathbb{C} \setminus \sigma_{\det}(a) : I(\det(a(S^1)) - \lambda, 0) \neq 0 \right\},
    \end{equation*} 
    and $B_0 \in \C^{k-1\times k-1}$.
\end{proposition}
The $k-1$ points mentioned in \cref{thm: GFBT} are the points in $\sigma(B_0)$, see \cite{ammari.barandun.ea2024Spectra} for details on $B_0$.
\begin{lemma}\label{lem:detform}
    Let $T(a)$ be a tridiagonal $k$-Toeplitz operator with symbol $a(z)$ as in \eqref{eq: k-Toeplitz base symbol}. Then, we have 
    \begin{equation*}
        \det (a(z) - \lambda) = \psi(z) + g(\lambda) \qquad \lambda,z \in \C,
    \end{equation*}
    where 
    \begin{equation}
        \psi(z)=(-1)^{k+1}\left((\prod_{i=1}^k c_i)z+(\prod_{i=1}^k b_i)z^{-1}\right),
    \end{equation}
    and $g(\lambda)$ is a polynomial\footnote{See \cite[Appendix A]{ammari.barandun.ea2024Spectra} for details.} of degree $k$.
\end{lemma}
As a consequence of the above two results, we can see that if we define $E$ to be the ellipse (with interior) traced out by $\psi(S^1)$, then we have 
\begin{align}
\label{eq: det and wind based on E}
    \sigma_{\det}(a) = (-g)^{-1}(\partial E) \quad \text{and} \quad \sigma_{\mathrm{wind}}(a) = (-g)^{-1}(\interior E).
\end{align}
The following results hold. 
\begin{lemma}\label{lem:ellipseprops}
    The map 
    \begin{align*}
        \psi:\mc B &\to E\\
        \qp &\mapsto \psi(e^{-\i(\qp)})
    \end{align*}
    is well-defined, surjective and for every $\xi \in E$, there exist unique $$\alpha \in [-\pi,\pi], \beta \in [0,\Delta/2]$$ such that
    \[
        \psi^{-1}(\xi) = \{\qp, \qpconj\}.
    \]

    Furthermore, if we denote by $E_\beta$ the (with interior) ellipse traced out by $\psi([-\pi,\pi]+\i\beta)$, then we have $E_\beta = E_{\Delta-\beta}$ and $\operatorname{int} E_{\beta_1} \supset E_{\beta_2}$ for $0\leq \beta_1<\beta_2 \leq \Delta/2$.
\end{lemma}
The proof of \cref{lem:ellipseprops} can be found in \cref{app:technical results}. An immediate consequence of this lemma is that for a given quasiperiodicity $\qp$ we have $\sigma(a(\qp)) = \sigma(a(\qpconj))$. We show graphically the statement of \cref{lem:ellipseprops} in \cref{fig:betacollapse}.

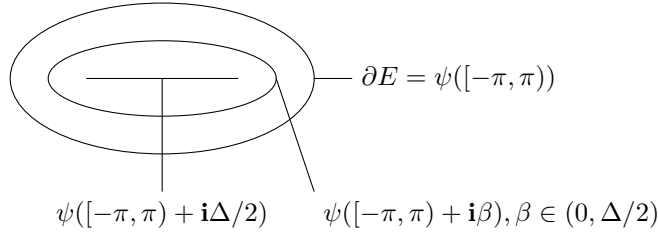
\begin{figure}[!h]\label{fig:betacollapse}
    \centering
    \begin{tikzpicture}
    % Draw ellipses
    \draw (0,0) ellipse (2 and 1);
    \draw (0,0) ellipse (1.5 and 0.5);

    \draw[-] (-1,0) -- (1,0);

    % Arrows and external labels
    \draw[-] (0,-1.5) node[anchor=north]{$\psi([-\pi,\pi)+\i\Delta/2)$} -- (0,0);
    \draw[-] (2.5,0) node[anchor=west]{$\partial E = \psi([-\pi,\pi))$} -- (2,0);
    \draw[-] (2,-1.5) node[anchor=north west]{$\psi([-\pi,\pi)+\i\beta), \beta \in (0,\Delta/2)$} -- (1.5,0);
    \end{tikzpicture}
    \caption{Parametrisation of the ellipse $E$ by $\psi$ for $\alpha\in Y^*\simeq S^1$ and $\beta\in [0,\Delta/2]$. As $\beta$ increases from $0$ to $\Delta/2$ the corresponding ellipse drawn out by $\psi(S^1+\i\beta)$ shrinks uniformly.}
\end{figure}

We are now ready to prove \cref{thm: GFBT}.
\begin{proof}[Proof of \cref{thm: GFBT}]
    We first remark that by \cref{prop: char of spectrum k toeplitz} we may consider $\sigma_{\det}(a) \cup \sigma_{\mathrm{wind}}(a)$ instead of $\sigma(T(a))$ by taking into account $(k-1)$ points that we may miss.

    We begin by proving the uniqueness statement. Let $\lambda \in \sigma(T(a))$. From \eqref{eq: det and wind based on E} we know that this is equivalent to $-g(\lambda)\in E$. By \cref{lem:ellipseprops}, there exist unique $\alpha\in [-\pi,\pi], \beta \in [0,\Delta/2]$ such that
    $\psi(\qp) = \psi(\qpconj) = -g(\lambda)$.
    This is equivalent to $\det(a(\qp)-\lambda)=\det(a(\qpconj)-\lambda)=0$ by \cref{lem:detform}, which also implies that $\qp$ is the unique quasiperiodicity such that $\lambda \in \sigma(a(\qp)) = \sigma(a(\qpconj))$. This also ensures $\sigma(T(a)) \subset \bigcup_{\qp \in \mc B} \sigma(a(\qp))$.
    
    To show $\sigma(T(a)) \supset \bigcup_{\qp \in \mc B} \sigma(a(\qp))$, we let $\lambda \in \sigma(a(\qp))$ for some $\qp \in \mc B$. After going backwards through the above argument, we find that this implies that $-g(\lambda) = \psi(\qp) \in E$. But, $-g(\lambda)\in E$ is equivalent to $\lambda \in \sigma_{\det}(a) \cup \sigma_{\mathrm{wind}}(a)$ as desired.
\end{proof}

At the beginning of this section, we have motivated the introduction of an imaginary part with the heuristic of taking into account the possible decay of the eigenvectors. This heuristic is made formal with the following proposition.
\begin{theorem}\label{prop:eigenvectors}
    Let $\lambda\in \sigma_{\det}(a) \cup \sigma_{\mathrm{wind}}(a)$, with the uniquely determined corresponding quasiperiodicities $\qp, \qpconj$.
    Let $\bm u_1, \bm u_2$ be eigenvectors of $a(\qp)$, $a(\qpconj)$ associated with that eigenvalue $\lambda$. Then, we can obtain a symbolic eigenvector $T(a) \bm u = \lambda \bm u$ of the Toeplitz operator as a linear combination of the $(\qp)$-quasiperiodic extension $\widetilde{\bm u}_1$ and the $(\qpconj)$-quasiperiodic extension $\widetilde{\bm u}_2$ of $\bm u_1, \bm u_2$, respectively. Furthermore, all eigenvectors $\bm u$ of $T(a)$ are of this form.
    Here, the $(\qp)$-quasiperiodic extension of a vector $\bm v$ is defined as 
    \[
        \widetilde{\bm v} \coloneqq (\bm v^\top, z^{-1}\bm v^\top, z^{-2}\bm v^\top, \dots)^\top
    \]
    for $z = e^{-\i(\qp)}$. Consequently, for every $j\in\N$,
    \begin{align}
        \frac{\vert \bm u^{(j+k)} \vert}{\vert \bm u^{(j)} \vert} = e^{-\beta}.
    \end{align}
\end{theorem}

\begin{proof}
    We first notice that $\widetilde{\bm u}_1$ and $\widetilde{\bm u}_2$ are linearly independent for $\qp \in \mc B$ with $\qp \neq \qpconj$ \cite{ammari.barandun.ea2024Spectra}. 
    
    Furthermore, we can see that both satisfy $T(a)\widetilde{\bm u}_i = \lambda \widetilde{\bm u}_i$ in all but the first row. We can thus find a linear combination of $\widetilde{\bm u}_1$ and $\widetilde{\bm u}_2$ which also satisfies the first row and is thus a \emph{proper} eigenvector of $T(a)$.
    For the case when $\qp = \qpconj$, which occurs if and only if $\beta = \Delta/2$ and $\alpha = 0$ or $\alpha = \pi$, we refer the reader to \cite[Proof of Theorem 2.9]{ammari.barandun.ea2024Spectra}.
    Finally, because the eigenspaces of tridiagonal operators with non-zero off-diagonal elements are at most one-dimensional (see \cref{lem:tridiagfd}), we know that all eigenvectors of $T(a)$ must take this form.
\end{proof}

\begin{remark}
    \cref{prop:eigenvectors} hence justifies why $\Delta$ is called the \emph{non-reciprocity rate}, as it directly translates into the decay rate of the eigenvector, a peculiarity of non-reciprocal systems. However, in the generic case an eigenvector $\bm u$ of $T(a)$ will be a linear combination of $\widetilde{\bm u}_1$ and $\widetilde{\bm u}_2$ with decay rates $\beta$ and $\Delta-\beta$ and thus has decay rate $\max\{\beta,\Delta-\beta\}$ which is maximised if $\beta=\Delta/2$. 
    Hence, the actual \emph{maximal rate of decay} is $\Delta/2$. 

    We can also see that the construction in \cref{prop:eigenvectors} is independent of the actual entries of $T(a)$ in its first row. Hence, it continues to work even if the top left edge of $T(a)$ is perturbed. 
\end{remark}

\section{Three Spectral Limits} \label{sec:3limits}
\begin{figure}
    \centering
    \includegraphics[width=\textwidth]{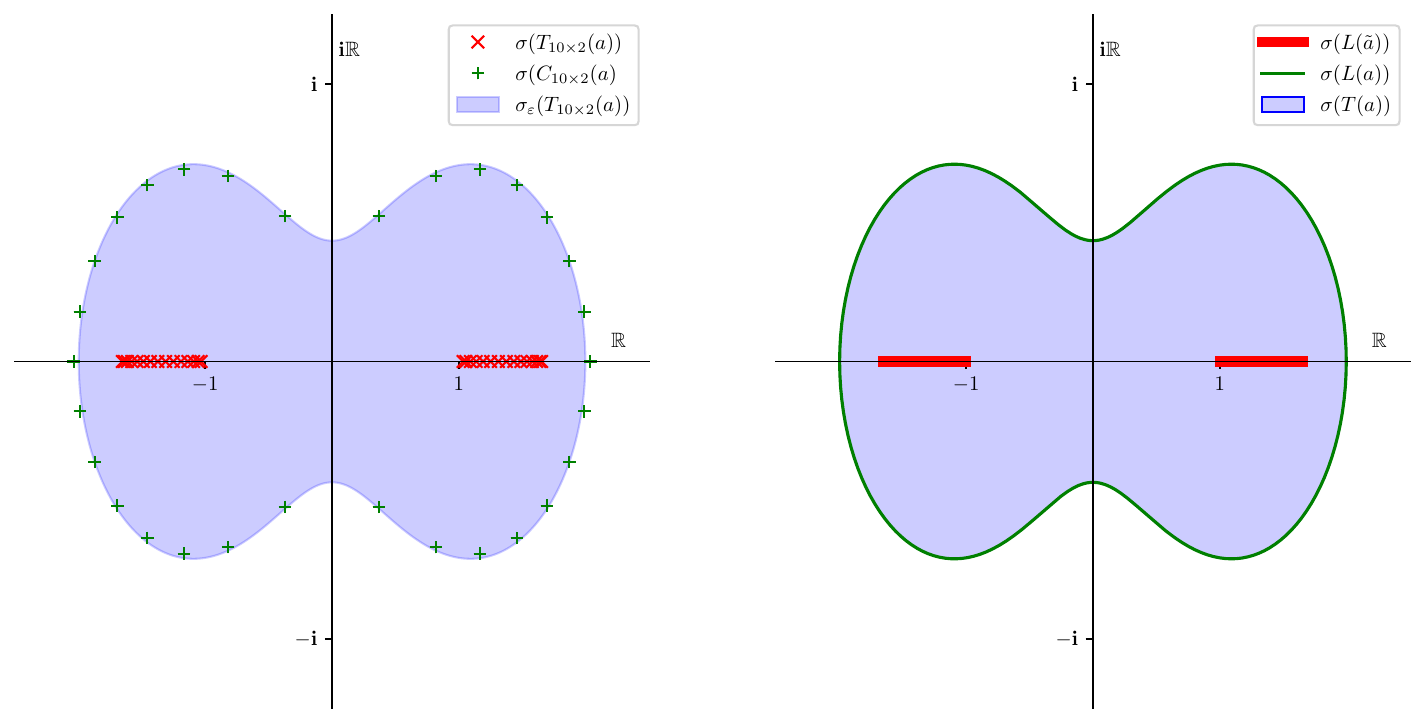}
    \caption{Illustration of the three spectral limits. We see that the eigenvalues of the circulant matrix $C_{2\times 10}(a)$ (green) arrange around the symbol curve and hence converge to the Laurent operator limit $L(a)$. The eigenvalues of the Toeplitz matrix $T_{2\times 10}(a)$ (red) arrange around the collapsed symbol $\Tilde{a}$ (as defined in the proof of \cref{thm:obcconv}) and hence converge to the collapsed Toeplitz (or Laurent) operator $T(\Tilde{a})$. The $\varepsilon$-pseudospectrum of $T_{mk}(a)$ corresponds exactly to the interior of the symbol curve and thus converges to the actual Toeplitz limit $T(a)$. }
    \label{fig:3lim}
\end{figure}

Having restored the Floquet-Bloch decomposition for the Toeplitz operator limit $T(a)$, we now aim to understand how and if finite Toeplitz matrices $T_{mk}(a)$ converge to this limit as $m\to \infty$. Crucially, while in the Hermitian setting the symbol $a(z)$ is collapsed, in general its eigenvalues trace out a path with nonempty interior. 
Any point in this interior is  exponentially close to a pseudoeigenvalue of $T_{mk}(a)$ in the limit $m\to \infty$.
This in turn causes the finite system and its limiting spectrum to be highly sensitive to boundary conditions. \cref{fig:3lim} shows the finite spectra under different boundary conditions as well as the pseudospectrum. The collapsed symbol causes their respective limits to coincide in the Hermitian setting but we can see that they diverge in the non-reciprocal setting.

The appropriate generalisation of the Brillouin zone depends on the limit of interest and does not necessarily match the generalised Brillouin zone for the Toeplitz operator limit $T(a)$ as defined in \cref{eq: GBZ}.
For open and periodic boundary conditions, we will now characterise the limiting spectrum and give the appropriate generalised Brillouin zone. Finally, we will investigate the limiting pseudospectrum and see how this connects the two boundary conditions.

\subsection{Open Limit} 
First, we aim to characterise the limiting spectrum $\sigma(T_{mk}(a))$ as $m\to \infty$, which corresponds to a growing finite system.
The main idea is that for non-reciprocal tridiagonal systems, we can perform a change of basis to obtain a similar symmetric system. The symbol of this system is collapsed and we can recover the traditional spectral convergence and decomposition. Then we may transform back into the original basis in order to achieve a Floquet-Bloch relation, but with a Brillouin zone shifted into the complex plane to account for the exponential decay of the eigenvectors.

\begin{theorem}\label{thm:obcconv}
    Let $T(a)$ be a tridiagonal Toeplitz operator with symbol $a(z)$ and $b_i, c_i\neq 0$ for all $1\leq i \leq k$. We then have
    \[
        \lim_{m\to \infty} \sigma(T_{mk}(a)) = \bigcup_{\alpha\in Y^*}\sigma(a(e^{-\i\spatialperiod(\alpha+\i\Delta/(2\spatialperiod))})).
    \]
\end{theorem}
\begin{proof}
    We begin by defining the diagonal change of basis $D_{mk}\in \C^{mk\times mk}$ with
    \[
        (D_{mk})_{ii} = \sqrt{\frac{T_{mk}(a)_{i-1,i}}{T_{mk}(a)_{i,i-1}}}(D_{mk})_{i-1,i-1} \quad \text{and} \quad (D_{mk})_{00}=1,
    \]
    \emph{i.e.}, the diagonal matrix given by the cumulative products of the off-diagonals of $T_{mk}(a)$. For the sake of simplicity, we assume that $b_ic_i>0$ for the rest of this proof. Nevertheless, note that the arguments can be extended to general nonzero off-diagonals after accounting for branch cuts of the square root. Using $D_{mk}$, we can now perform a change of basis and find that $T(\Tilde{a})_{mk}\coloneqq D_{mk}T_{mk}(a)D_{mk}^{-1}$ is symmetric and remains tridiagonal $k$-Toeplitz, justifying our notation and making $\Tilde{a}$ a well-defined symbol. This can be seen by comparing the entries and holds even if $D_{mk}$ or $T_{mk}(a)$ contain complex entries. While $T(\Tilde{a})_{mk}$ is symmetric, it is in general non-Hermitian or even not normal, as it can contain complex values. However, being symmetric ensures that the symbol $\Tilde{a}$ is collapsed and allows us to apply classical Toeplitz theory of convergence (namely \cref{eq:ClassicalFBT,eq:ClassicalConv}). In this case, we get
    \[
         \lim_{m\to \infty}\sigma(T_{mk}(a)) = \lim_{m\to \infty}\sigma(T_{mk}(\Tilde{a})) = \sigma(L(\Tilde{a})) = \bigcup_{\alpha\in Y^*}\sigma(\Tilde{a}(e^{-\i \spatialperiod\alpha})).
    \]
    It remains to show that $\sigma(\Tilde{a}(e^{-\i \spatialperiod\alpha}))=\sigma(a(e^{-\i\spatialperiod(\alpha+\i\Delta/2L)}))$. This follows from the fact that 
    $D_{1k}^{-1}\Tilde{a}(e^{-\i \spatialperiod\alpha})D_{1k} = a(e^{-\i \spatialperiod(\alpha+\i\Delta/(2\spatialperiod))}).$
\end{proof}

This result is illustrated in \cref{fig:3lim} where the spectrum of the Toeplitz matrix on the left-hand side (corresponding to the open boundary condition) converges to the spectrum of the Toeplitz operator with collapsed symbol $T(\Tilde{a})$ on the right-hand side. As we just proved the appropriate \emph{generalised Brillouin zone} to decompose the spectrum of this operator is the classical Brillouin zone, shifted into the complex plane:
\[
\mc B_\text{OBC} = \bigr\{ \alpha + \i\Delta/(2\spatialperiod) \mid \alpha\in Y^* \bigr\}.
\]

Consequently, for tridiagonal Toeplitz systems with open boundary conditions, shifting the Brillouin zone by $\Delta/(2\spatialperiod)$ restores spectral convergence as well as the Floquet-Bloch decomposition. This corresponds to the fact that in the tridiagonal case, all the eigenmodes have the same rate of decay. Namely, this decay is the maximal decay that is given explicitly by $\Delta/(2\spatialperiod)$.

\subsection{Periodic boundary conditions and the Laurent operator limit}
We can impose periodic boundary conditions on $T_{mk}(a)$ and  get the tridiagonal $k$-circulant matrix
\begin{equation}
    (C_{mk}(a))_{ij} \coloneqq \begin{cases}
        c_k & i=0, j=mk,\\
        b_k & i=mk, j=0,\\
        (T_{mk}(a))_{ij} & \text{otherwise}.\\
    \end{cases}
\end{equation}

The following result holds. 
\begin{theorem}\label{thm:pbcconv}
Let $C_{mk}(a)$ be a tridiagonal $k$-circulant matrix as above. Then, we have the following spectral decomposition:
    \[
        \sigma(C_{mk}(a)) = \bigcup_{j=0}^{m-1}\sigma(a(e^{2\pi\i j/m})).
    \]
Furthermore, if we let $m \to \infty$,
\begin{equation}\label{equ:periodiclimiting}
    \sigma(C_{mk}(a)) = \bigcup_{j=0}^{m-1}a(e^{2\pi\i j/m}) \to \bigcup_{\alpha\in Y^*}a(e^{-\i\spatialperiod \alpha}) = \sigma(L(a)).
\end{equation}
\end{theorem}
\begin{proof}
    We begin by proving the first equality.
    The right inclusion follows from the fact that extending any eigenvector of $a(e^{2\pi\i j/m})$ quasiperiodically yields an eigenvector of $C_{mk}(a)$.
    The left inclusion follows from multiplicity, as the right-hand side yields $m\times k$ eigenvalues. The proposition continues to hold even in the case where some $a(e^{2\pi\i j/m})$ might not be diagonalisable, since the same argument can be carried out for generalised eigenvectors.

    The second equality then follows from the fact that $\{e^{2\pi\i j/m}\mid j\in 0,\dots,m\}\to S^1$ (in the Hausdorff sense) as $m\to \infty$ and equation \cref{eq:ClassicalFBT}.
\end{proof}

We have thus shown the analogue of \cref{thm:obcconv} for periodic boundary conditions. As we can see in \cref{fig:3lim}, imposing periodic boundary conditions on the finite system causes the spectrum of $C_{mk}(a)$ to diverge drastically from the spectrum of $T_{mk}(a)$. This corresponds to the fact that while the eigenmodes of $T_{mk}(a)$ have a decay of $\Delta/2L$, imposing periodic boundary conditions forces the eigenmodes to be decay-free, causing a large perturbation. The non-Hermitian non-reciprocity thus causes the system to be highly sensitive to boundary conditions.
\cref{fig:3lim} further illustrates how the spectrum of the circulant matrices $C_{mk}(a)$ arranges around the symbol curve and thus converges to the Laurent operator limit as $m\to \infty$. The above theorem therefore shows that the appropriate Brillouin zone for this setting is the classical Brillouin zone with no decay: $$\mc B_\text{PBC}=Y^*.$$ 

\subsection{Pseudospectra and the Toeplitz operator limit}
Finally, we investigate the pseudospectrum of the finite system $T_{mk}(a)$. Crucially, while the spectrum of Toeplitz matrices is highly sensitive to boundary conditions, the pseudospectrum is not and using \cite[Corollary 6.16]{bottcher.silbermann1999Introduction},  we find that it converges to the Toeplitz operator limit.
\begin{theorem}[\citeauthor{bottcher.silbermann1999Introduction}]
    Consider a continuous symbol $a\in L^\infty(S^1,\C^{k\times k})$ so that $T(a)$ is tridiagonal. Then, for every $\epsilon > 0$,
\begin{align}
    \lim_{m\to\infty} \sigma_\epsilon(T_{mk}(a)) = \sigma_\epsilon (T(a)).
\end{align}
\end{theorem}
In particular, the previous theorem implies that
$$
\lim_{\epsilon \to 0}\lim_{m\to\infty}\sigma_\epsilon (T_{mk}(a)) = \lim_{\epsilon \to 0}\sigma_\epsilon (T(a))   = \sigma (T(a)).
$$

Hence, by \cref{thm: GFBT}, the appropriate generalised Brillouin zone to recover the pseudospectral limit is the same as for the Toeplitz operator: 
\[
\mc B = \bigg\{\qp \mid \alpha\in [-\pi/\spatialperiod,\pi/\spatialperiod), \beta \in [0,\Delta/\spatialperiod] \bigg\}.
\] 
Notably, this Brillouin zone is of higher dimension than the previous two Brillouin zones (\emph{i.e.} a two dimensional region in $\C$ different from the previous $S^1\cong Y^*$) as it contains a range of possible decay rates. Furthermore, it contains the shifted Brillouin zone of the open boundary condition and the classical Brillouin zone of the periodic boundary condition as special cases ($\beta = \Delta/(2\spatialperiod)$ and $\beta=0$). This is in line with the fact that the Toeplitz operator spectrum contains both the Laurent operator spectrum, as well as the collapsed Toeplitz spectrum, as seen in \cref{fig:3lim}. The open and periodic boundary condition settings thus correspond to the maximal decay and zero decay extremes of a range of possible pseudospectral decays, captured by the Toeplitz operator spectrum.

\section*{Acknowledgements}
The authors would like to thank Erik Orvehed Hiltunen for insightful discussions. This work was supported by Swiss National Science Foundation grant number 200021--200307. 

\appendix

\section{Technical Results}\label{app:technical results}
\begin{proof}[Proof of \cref{lem:ellipseprops}]
    We denote $p=\Arg (\prod_{j=1}^k c_j), q = \Arg (\prod_{j=1}^k b_j)$ and find
    \begin{align*}
        \psi(z) =& (-1)^{k+1}\left((\prod_{j=1}^k c_j)z+(\prod_{j=1}^k b_j)z^{-1}\right)\\
        =&
        (-1)^{k+1}\left((\prod_{j=1}^k \abs{c_j})e^{\i p}z+(\prod_{j=1}^k \abs{b_j})e^{\i q}z^{-1}\right) \\
        =& (-1)^{k+1}\left((\prod_{j=1}^k \abs{c_j})e^{\i(p+q)/2}e^{-\i(q-p)/2}z+(\prod_{j=1}^k \abs{b_j})e^{\i(p+q)/2}e^{\i(q-p)/2}z^{-1}\right) 
        \\
        =& (-1)^{k+1}e^{\i(p+q)/2}\left((\prod_{j=1}^k \abs{c_j})(ze^{-\i(q-p)/2})+(\prod_{j=1}^k \abs{b_j})(ze^{-\i(q-p)/2})^{-1}\right) 
        \\
        =& K\Tilde{\psi}(ze^{-\i\zeta/2}),
    \end{align*}
    where $K\coloneqq (-1)^{k+1}\prod_{j=1}^k\sqrt{\frac{b_jc_j}{\abs{b_jc_j}}}$, $\zeta \coloneqq \Arg\left(\prod_{j=1}^k\frac{b_j}{c_j}\right)$ and
    \[
        \Tilde{\psi}(z) = \underbrace{(\prod_{j=1}^k \abs{c_j})}_{\eqqcolon A^-}z + \underbrace{(\prod_{j=1}^k \abs{b_j})}_{\eqqcolon A^+}z^{-1}.
    \]
    We can now study the ellipse traced out by $\Tilde{\psi}(z)$ which is unrotated and centred at $0$. If we respectively denote by $a,b$ its semi-major and semi-minor axis,  then we find that $A^\pm = \frac{a\pm b}{2}$. Note that $\Delta>0$ implies $b=A^+-A^->0$. 
    If we now allow $\beta \in [0,\Delta]$, we find the ellipse
    \[
        \Tilde{\psi}(e^{-\i(\alpha+\i\beta)}) = A^-e^\beta e^{-\i\alpha} + A^+e^{-\beta} e^{+\i\alpha}.
    \]
    The semi-minor axis of this ellipse is thus given by 
    \[
        b = \frac{A^+e^{-\beta} - A^-e^{\beta}}{2} = \frac{1}{2}\left((\prod_{j=1}^k \abs{b_j})^{1-\beta'}(\prod_{j=1}^k \abs{c_j})^{\beta'} - (\prod_{j=1}^k \abs{c_j})^{1-\beta'}(\prod_{j=1}^k \abs{b_j})^{\beta'}\right),
    \]
    where in the last equality we have factored out $\beta'\Delta = \beta$ and used $ \Delta = \ln \prod_{j=1}^k \abs{\frac{b_j}{c_j}}$. We can see that $b$ decreases from $(A^+-A^-)/2$ to $0$ as $\beta'$ increases from $0$ to $1/2$ and reaches $-(A^+-A^-)/2$ for $\beta'=1$. Thus, $\Tilde{\psi}(e^{-\i(\alpha+\i\beta)}), \alpha\in [-\pi,\pi]$  shrinks uniformly as $\beta$  increases from $0$ to $\Delta/2$ and so must $\psi$. Restricting $\beta$ to $[0,\Delta]$ also ensures that $\psi$ is well-defined, \emph{i.e.}, $\psi(e^{-\i(\alpha+\i\beta)})\in E$ and that $\psi$ is subjective. 
    We can also see that $b\mapsto -b$ if $\beta'\mapsto1-\beta'$ which implies $\Tilde{\psi}(e^{-\i(\alpha+\i\beta)}) = \Tilde{\psi}(e^{-\i(-\alpha+\i(\Delta-\beta))})$. Using this fact, we find that
    \begin{gather*}
        \psi(e^{-\i(\alpha+\i\beta)}) = K\Tilde{\psi}(e^{-\i(\alpha+\zeta/2+\i\beta)}) = K\Tilde{\psi}(e^{-\i(-\zeta/2-\alpha+\i(\Delta-\beta))}) \\= \psi(e^{-\i(-\zeta-\alpha+\i(\Delta-\beta))}).
    \end{gather*}
     Since $\psi(z)=\xi$ is a quadratic equation, it follows that
     $\psi^{-1}(\xi) = \bigr\{ \qp, (-\zeta - \alpha) + \i(\Delta - \beta) \bigr\}$. 
\end{proof}

\begin{lemma}\label{lem:tridiagfd}
    Let $A$ be a tridiagonal operator on $\ell(\C)$ or a matrix on $\C^{n\times n}$ with non-zero entries in the off-diagonals. Then its eigenspaces have at most dimension one.
\end{lemma}
\begin{proof}
We consider the following tridiagonal operator on the (potentially unbounded) sequence space $\ell(\C)$:
    \begin{align*}
    A = \begin{pmatrix}
    a_1 & b_1 &0  &  \cdots  \\
            c_1    & a_2 & b_2 & \\
            0      & c_2 & \ddots & \ddots \\
            \vdots    & & \ddots & \ddots
    \end{pmatrix}.
    \end{align*}
    Suppose that there exists an eigenvalue $\lambda\in \sigma(A)$ with corresponding eigenvector $\bm u = (u_1,\dots)\in \ell(C)$. Then, the entries of $\bm u$ satisfy the following set of equations: $u_2 = -(a_1-\lambda)u_1/b_2$ and $u_{i+1} = -(c_{i-1}u_{i-1}+(a_i-\lambda)u_i)/b_{i+1}$ for $i=2,\dots$. This recursive relation thus uniquely determines $\bm u$ from $u_1$. 
    Let now $\bm u'$ be another eigenvector to $\lambda$. Because the recursive relation is linear, picking $c\coloneqq u_1'/u_1$ yields $\bm u' = c\bm u$ and $\bm u', \bm u$ must be linearly dependent. Because $\bm u'$ was arbitrary, this proves that the eigenspace can have at most dimension one. 

    This argument functions analogously for the matrix case.
\end{proof}

\printbibliography

\end{document}